\documentclass{style}

\usepackage{amsmath, amssymb, graphicx, subfigure, verbatim}
\usepackage[hyphens]{url}
\usepackage{hyperref}
\usepackage{xspace}

\newcommand{\complexityclass}[1]{\textbf{#1}}
\newcommand{\computproblem}[1]{\textsc{#1}}
\renewcommand{\P}{\complexityclass{P}\xspace}
\newcommand{\NP}{\complexityclass{NP}\xspace}

\newcommand{\TEXP}{\complexityclass{2-EXP}\xspace}
\newcommand{\PSPACE}{\complexityclass{PSPACE}\xspace}
\newcommand{\NPSPACE}{\complexityclass{NPSPACE}\xspace}

\newcommand{\EENCL}{\computproblem{EE-NCL}\xspace}
\newcommand{\EEANCL}{\computproblem{EE-ANCL}\xspace}
\newcommand{\SSP}{\computproblem{SSP}\xspace}
\newcommand{\PSSP}{\computproblem{PSSP}\xspace}

\title{Partial Searchlight Scheduling is Strongly \PSPACE-Complete}

\author{Giovanni Viglietta\thanks{School of Computer Science, Carleton University, Ottawa ON, Canada, {\tt viglietta@gmail.com}.}}

\index{Viglietta, Giovanni}

\begin{document}
\thispagestyle{empty}
\maketitle

\begin{abstract}
The problem of searching a polygonal region for an unpredictably moving intruder by a set of stationary guards, each carrying an orientable laser, is known as the \computproblem{Searchlight Scheduling Problem}. Determining the computational complexity of deciding if the polygon can be searched by a given set of guards is a long-standing open problem.

Here we propose a generalization called the \computproblem{Partial Searchlight Scheduling Problem}, in which only a given subregion of the environment has to be searched, as opposed to the entire area. We prove that the corresponding decision problem is strongly \PSPACE-complete, both in general and restricted to orthogonal polygons where the region to be searched is a rectangle.

Our technique is to reduce from the ``edge-to-edge'' problem for \emph{nondeterministic constraint logic machines}, after showing that the computational power of such machines does not change if we allow ``asynchronous'' edge reversals (as opposed to ``sequential'').
\end{abstract}

\section{Introduction} \label{sectionintro}

\paragraph{Previous work.}

The \computproblem{Searchlight Scheduling Problem} (\SSP), first studied in~\cite{search}, is a pursuit-evasion problem in which a polygon has to be searched for a moving intruder by a set of stationary guards. The intruder moves unpredictably and continuously with unbounded speed, and each guard carries an orientable \emph{searchlight}, emanating a 1-dimensional ray that can be continuously rotated about the guard itself. The polygon's exterior cannot be traversed by the intruder, nor penetrated by searchlights. The intruder is caught whenever it is hit by a searchlight. Because the intruder's location is unknown until it is actually caught, each guard has to sway its searchlight according to a predefined \emph{schedule}. If the guards always catch the intruder, regardless of its path, by following their schedules in concert, they are said to have a \emph{search schedule}.

\SSP is the problem of deciding if a given set of guards has a search schedule for a given polygon (possibly with holes). The computational complexity of this decision problem has been only marginally addressed in~\cite{search}, but has later gained more attention, until in~\cite{bullo} the space of all possible schedules has been shown to be discretizable and reducible to a finite graph, which can be explored systematically to find a search schedule, if one exists. Since the graph may have double exponential size, this technique easily places \SSP in \TEXP. Whether \SSP is \NP-hard or even in \NP is left in~\cite{bullo} as an open problem.

More recently, in~\cite{viglietta2,viglietta}, the author studied the complexity of a 3-dimensional version of \SSP, in which  the input polygonal environment is replaced by a polyhedron, and the 1-dimensional rays become 2-dimensional half-planes, which rotate about their boundary lines. This variation of \SSP is shown to be strongly \NP-hard.

\paragraph{Our contribution.}

In the present paper we take a further step along this line of research, by introducing the \computproblem{Partial Searchlight Scheduling Problem} (\PSSP), in which the guards content themselves with searching a smaller subregion given as input. That is, a search schedule should only guarantee that the given \emph{target region} is eventually cleared, either by catching the intruder or by confining it outside. We prove that \PSSP is strongly \PSPACE-complete, both for general polygons and restricted to orthogonal polygons in which the region to be searched is a rectangle.

To prove that \PSSP is a member of \PSPACE, we do a refined analysis of the discretization technique of~\cite{bullo}. To prove \PSPACE-hardness, we give a reduction from the ``edge-to-edge'' problem for nondeterministic constraint logic machines, discussed in~\cite{ncl}. Another contribution of this paper is the observation that the nondeterministic constraint logic model of computation stays essentially the same if we allow ``asynchronous'' moves, as opposed to ``sequential'' ones.

An earlier version of this paper has appeared in \cite{eurocg}, and most of the material is also contained in the author's Ph.D. thesis \cite{thesis}.

\section{Preliminary observations}

An instance of \PSSP is a triplet $(\mathcal P,G,\mathcal T)$, where $\mathcal P$ is a polygon, possibly with holes, $G$ is a finite set of point guards located in $\mathcal P$ or on its boundary, and $\mathcal T \subseteq \mathcal P$ is a target polygonal region. The question is whether the guards in $G$ can turn their lasers in concert, from a fixed starting position and following a finite schedule, so as to guarantee that in the end any intruder that moves in $\mathcal P$ and tries to avoid lasers is necessarily not in $\mathcal T$.

We remark that $\SSP \preceq_{\P} \PSSP$ trivially, in that in \SSP we have $\mathcal T=\mathcal P$ always. One feature of \SSP that is not preserved by this generalization is what we call the \emph{time reversal invariance} property. In \SSP, a given schedule successfully searches $\mathcal P$ if and only if reversing it with respect to time also searches $\mathcal P$. In contrast, this is not the case with \PSSP, and Figure~\ref{fig0:a} shows a simple example. The dark target region can be cleared only if the guard turns its searchlight clockwise, as indicated by the arrow. If the searchlight is turned counterclockwise instead, the intruder can first hide in the protected area on the left, then come out and safely reach the target region. Protected areas like this one, that cannot be searched because they are invisible to all guards, provide a constant source of \emph{recontamination}, and will be extensively used in our main \PSPACE-hardness reduction (see Lemma~\ref{lemma2}).

\begin{figure}[h]
\centering
\subfigure[]{\label{fig0:a}\includegraphics[scale=.65]{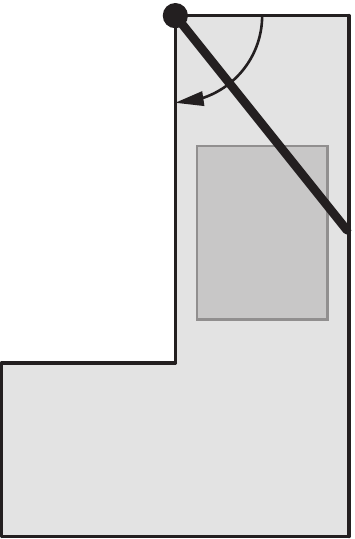}}\qquad\qquad\quad
\subfigure[]{\label{fig0:b}\includegraphics[scale=.65]{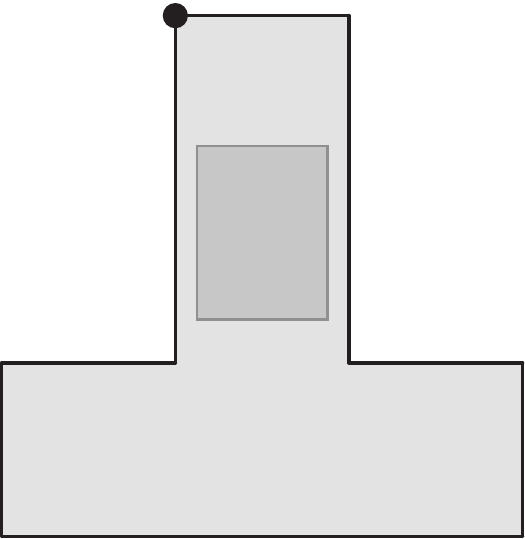}}
\caption{Two instances of \PSSP}
\label{fig0}
\end{figure}

A search heuristic called \emph{one-way sweep strategy} was described in~\cite{search} for \SSP restricted to simple polygons, and later extended to polygons with holes in~\cite{search2}. An interesting consequence of this heuristic is that, if a set of guards lies on the boundary of a simple polygon, and no point in the polygon is invisible to all guards, then there is a schedule that successfully searches the whole polygon. However, such property does not straightforwardly generalize to \PSSP, as Figure~\ref{fig0:b} illustrates. Here we have a simple polygon with a guard on the boundary that can see the whole target region. But, because there are protected areas on both sides, the target region is unsearchable, no matter in which direction the guard turns it searchlight.

As it turns out, adding just another guard anywhere on the boundary of the polygon in Figure~\ref{fig0:b} makes the target region searchable. For example, if the second guard is placed on the top-right corner, it can orient its laser downward, thus ``closing'' the right protected area, and allowing the first guard to search the target region as in Figure~\ref{fig0:a}. On the other hand, if the second guard is placed on the boundary of any protected area, then the whole polygon is visible to the guards, and therefore we know that it can be entirely searched.

\section{NCL machines and asynchrony} \label{sectionncl}

Our \PSPACE-hardness reduction is based on a model of computation called nondeterministic constraint logic, whose definition and main properties are detailed in~\cite{ncl}. Here we extend the basic model by introducing \emph{asynchrony}, and showing that its computational power stays the same.

\paragraph{Basic NCL machines.}

Consider an undirected 3-connected 3-regular planar graph, whose vertices can be of two types: \emph{AND vertices} and \emph{OR vertices}. Of the three edges incident to an AND vertex, one is called its \emph{output edge}, and the other two are its \emph{input edges}. Such a graph is (a special case of) a \emph{nondeterministic constraint logic machine (NCL machine)}. A \emph{legal configuration} of an NCL machine is an orientation (direction) of its edges, such that:
\begin{itemize}
\item for each AND vertex, either its output edge is directed inward, or both its input edges are directed inward;
\item for each OR vertex, at least one of its three incident edges is directed inward.
\end{itemize}
A \emph{legal move} from a legal configuration to another configuration is the reversal of a single edge, in such a way that the above constraints remain satisfied (i.e., such that the resulting configuration is again legal).

Given an NCL machine with two \emph{distinguished edges} $e_a$ and $e_b$, and a \emph{target orientation} for each, we consider the problem of deciding if there exist legal configurations $A$ and $B$ such that $e_a$ has its target orientation in $A$, $e_b$ has its target orientation in $B$, and there is a sequence of legal moves from $A$ to $B$. In a sequence of moves, the same edge may be reversed arbitrarily many times. We call this problem \computproblem{Edge-to-Edge for Nondeterministic Constraint Logic machines} (\EENCL).

A proof that \EENCL is \PSPACE-complete is given in~\cite{ncl}, by a reduction from \computproblem{True Quantified Boolean Formula}. By inspecting that reduction, we may further restrict the set of \EENCL instances on which we will be working. Namely, we may assume that $e_a\neq e_b$, and that in no legal configuration both $e_a$ and $e_b$ have their target orientation.

\paragraph{Asynchrony.}

For our main reduction, it is more convenient to employ an \emph{asynchronous} version of \EENCL. Intuitively, instead of ``instantaneously'' reversing one edge at a time, we allow any edge to start reversing at any given time, and the \emph{reversal phase} of an edge is not ``atomic'' and instantaneous, but may take any strictly positive amount of time. It is understood that several edges may be in a reversal phase simultaneously. While an edge is reversing, its orientation is undefined, hence it is not directed toward any vertex. During the whole process, at any time, both the above constraints on AND and OR vertices must be satisfied. We also stipulate that no edge is reversed infinitely many times in a bounded timespan, or else its orientation will not be well-defined in the end. With these extended notions of configuration and move, and with the introduction of ``continuous time'', \EENCL is now called \computproblem{Edge-to-Edge for Asynchronous Nondeterministic Constraint Logic machines} (\EEANCL).

Despite its asynchrony, such new model of NCL machine has precisely the same power of its traditional synchronous counterpart.

\begin{theorem}\label{asynch}$\EENCL = \EEANCL$.\end{theorem}

\begin{proof}
Obviously $\EENCL \subseteq \EEANCL$, because any sequence of moves in the synchronous model trivially translates into an equivalent sequence for the asynchronous model.

For the opposite inclusion, we show how to ``serialize'' a legal sequence of moves for an asynchronous NCL machine going from a legal configuration $A$ to configuration $B$ in a bounded timespan, in order to make it suitable for the synchronous model. An asynchronous sequence is represented by a set
$$S=\{(e_m, s_m, t_m) \mid m\in M\},$$ where $M$ is a set of ``edge reversal events'', $e_m$ is an edge with a reversal phase starting at time $s_m$ and terminating at time $t_m > s_m$. For consistency, no two reversal phases of the same edge may overlap.

Because no edge can be reversed infinitely many times, $S$ must be finite. Hence we may assume that $M=\{1, \cdots, n\}$, and that the moves are sorted according to the (weakly increasing) values of $s_m$, i.e., $1\leqslant m < m' \leqslant n \implies s_m \leqslant s_{m'}$. Then we consider the serialized sequence
$$S'=\{(e_m, m, m) \mid m\in M\},$$
and we claim that it is valid for the synchronous model, and that it is equivalent to $S$.

Indeed, each move of $S'$ is instantaneous and atomic, no two edges reverse simultaneously, and every edge is reversed as many times as in $S$, hence the final configuration is again $B$ (provided that the starting configuration is $A$). We still have to show that every move in $S'$ is legal. Let us do the first $m$ edge reversals in $S'$, for some $m\in M$, starting from configuration $A$, and reaching configuration $C$. To prove that $C$ is also legal, consider the configuration $C'$ reached in the asynchronous model at time $s_m$, according to $S$, right when $e_m$ starts its reversal phase (possibly simultaneously with other edges). By construction of $S'$, every edge whose direction is well-defined in $C'$ (i.e., every edge that is not in a reversal phase) has the same orientation as in $C$. It follows that, for each vertex, its inward edges in $C$ are a superset of its inward edges in $C'$. By assumption on $S$, $C'$ satisfies all the vertex constraints, then so does $C$, {\it a fortiori}.
\end{proof}

\begin{cor}\label{eeancl}\EEANCL is \PSPACE-complete. \hfill $\square$\end{cor}

%\begin{proof}
%Recall from~\cite{ncl} that \EENCL\ is \PSPACE-complete, and that $\EENCL = \EEANCL$ by Proposition~\ref{asynch}.
%\end{proof}

\section{\PSPACE-completeness of \PSSP}

To prove that \PSSP belongs to \PSPACE we use the discretization technique of~\cite{bullo}, and to prove that \PSSP is \PSPACE-hard we give a reduction from \EEANCL.

\paragraph{Membership.}

Due to Savitch's theorem, it suffices to show that \PSSP belongs to \NPSPACE.

\begin{lemma} \label{lemma1}
\PSSP $\in$ \NPSPACE.
\end{lemma}
\begin{proof}
As detailed in~\cite{bullo}, a technique known as \emph{exact cell decomposition} allows to reduce the space of all possible schedules to a finite graph $G$. Each searchlight has a linear number of \emph{critical angles}, which yield an overall partition of the polygon into a polynomial number of \emph{cells}. In the discretized search space, searchlights take turns moving, and can stop or change direction only at critical angles. Thus, a vertex of $G$ encodes the \emph{status} of each cell (either \emph{contaminated} or \emph{clear}) and the critical angle at which each searchlight is oriented.

As a consequence, $G$ can be navigated nondeterministically by just storing one vertex at a time, which requires polynomial space. Notice that deciding if two vertices of $G$ are adjacent can be done in polynomial time: an edge in $G$ represents a move of a single searchlight between two consecutive critical angles, and the updated status of each cell can be easily evaluated. Indeed, cells' vertices are intersections of lines through input points, hence their coordinates can also be efficiently stored and handled as rational expressions involving the input coordinates.

Now, in order to verify that a path in $G$ is a witness for \SSP, one checks if the last vertex encodes a status in which every cell is clear. But the very same cell decomposition works also for \PSSP: the analysis in~\cite{bullo} applies even if just a subregion of the polygon has to be searched, and a path in $G$ is a witness for \PSSP if and only if its last vertex encodes a status in which every cell that has a non-empty intersection with the target subregion is clear.
\end{proof}

\paragraph{Hardness.}

For the \PSPACE-hardness part, we first give a reduction in which the target region to be cleared is an orthogonal hexagon. Then, in Section~\ref{convexregion} we will explain how to modify our construction, should we insist on having a rectangular (hence convex) target region.

\begin{lemma} \label{lemma2}
\EEANCL $\preceq_{\P}$ \PSSP restricted to orthogonal polygons.
\end{lemma}
\begin{proof}
We show how to transform a given asynchronous NCL machine $G$ with two distinguished edges $e_a$ and $e_b$ into an instance of \PSSP.

A rough sketch of our construction is presented in Figure~\ref{fig1}. All the vertices of $G$ are placed in a row~(a), and are connected together by a network of thin \emph{corridors}~(b), turning at right angles, representing edges of $G$. (Although $G$ is 3-regular, only a few of its edges are sketched in Figure~\ref{fig1}.)
Each \emph{subsegment} of a corridor is a thin rectangle, containing a \emph{subsegment guard} in the middle (not shown in Figure~\ref{fig1}). Two subsegments from different corridors may indeed cross each other like in~(c), but in such a way that the crossing point is far enough from the ends of the two subsegments and from the two subsegment guards (so that no subsegment guard can see all the way through another subsegment). All the vertices of $G$ and all the \emph{joints} between consecutive subsegments (i.e., the turning points of each corridor) are connected via extremely thin \emph{pipes}~(d) to the upper area~(e), which contains the target region (shaded in Figure~\ref{fig1}).

\begin{figure}[ht]
\centering
\includegraphics[width=\linewidth]{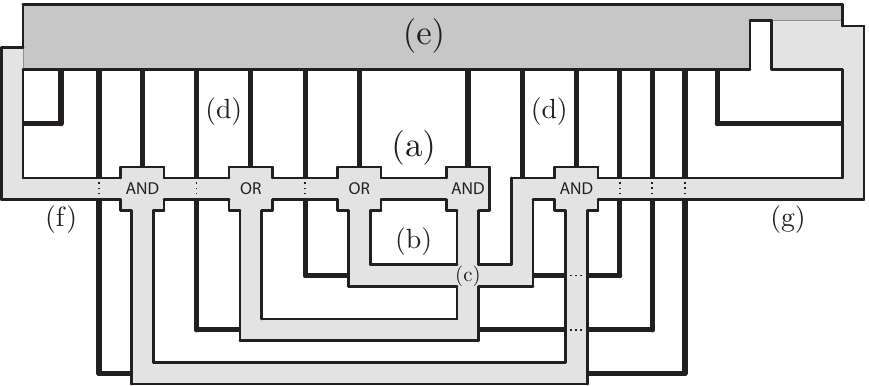}
\caption{Construction overview}
\label{fig1}
\end{figure}

Two corridors~(f) and~(g) also reach the upper area, and they correspond to the distinguished edges of $G$, $e_a$ and $e_b$, respectively. That is, if $e_a=\{u,v\}$, and the target orientation of $e_a$ is toward $v$, then the corridor corresponding to $e_a$ connects vertex $u$ in our construction to the upper area~(e), rather than to $v$. The same holds for $e_b$. Indeed, observe that we may assume that $e_a$ and $e_b$ are reversed only once (respectively, on the first and last move) in a sequence of moves that solves \EEANCL on $G$. As a consequence, contributions to vertex constraints given by distinguished edges oriented in their target direction may be ignored.

Each pipe turns at most once, and contains one \emph{pipe guard} in the middle, lying on the boundary. Notice that straight pipes never intersect corridors, but some turning pipes do. Figure~\ref{fig2} shows a turning pipe, with its pipe guard~(a) and an intersection with a corridor~(b) (proportions are inaccurate). The \emph{intersection guards}~(c) separate the pipe from the corridor with their lasers (dotted lines in Figure~\ref{fig2}), without ``disconnecting'' the pipe itself. Although a pipe narrows every time it crosses a corridor, its pipe guard can always see all the way through it, because it is located in the middle. The small \emph{nook}~(d) is unclearable because no guard can see its bottom, hence it is a constant source of recontamination for the target region~(e), unless the pipe guard is covering it with its laser. (Each straight pipe also has a similar nook.)

\begin{figure}[h]
\centering
\includegraphics[width=\linewidth]{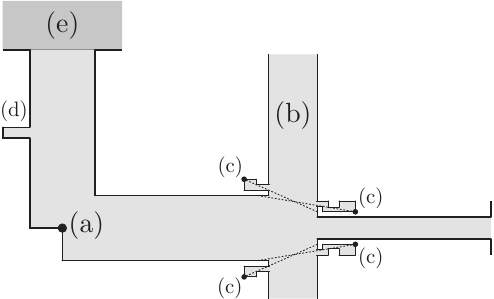}
\caption{Intersection between a pipe and a corridor}
\label{fig2}
\end{figure}

In our construction, corridor guards implement edge orientations in $G$: whenever all the subsegment guards in a corridor connecting vertices $u$ and $v$ have their lasers oriented in the same ``direction'' from vertex $u$ to vertex $v$, it means that the corresponding edge $\{u,v\}$ in $G$ is oriented toward $v$.

Figure~\ref{fig3} shows an OR vertex. The three subsegment guards from incoming corridors~(a) can all ``cap'' pipe~(b) with their lasers, and nook~(c) guarantees that the pipe is recontaminated whenever all three guards turn their lasers away.

\begin{figure}[h]
\centering
\includegraphics[width=\linewidth]{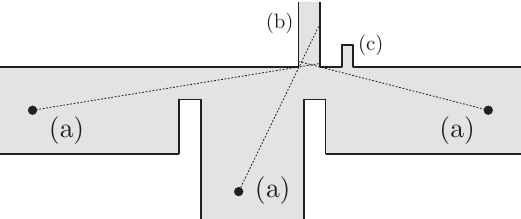}
\caption{OR vertex}
\label{fig3}
\end{figure}

AND vertices are implemented as in Figure~\ref{fig4}. The two subsegment guards~(a) correspond to input edges, and are able to cap one pipe~(e) each, whereas guard~(c) can cover them both simultaneously. But that leaves pipe~(d) uncovered, unless it is capped by guard~(b), which belongs to the corridor corresponding to the output edge. Again, uncovered pipes are recontaminated by unclearable nooks~(f).

\begin{figure}[h]
\centering
\includegraphics[width=\linewidth]{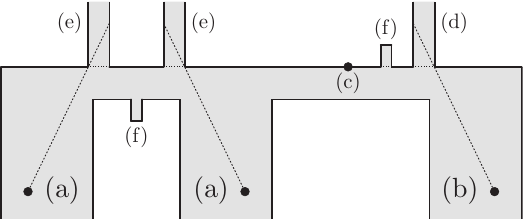}
\caption{AND vertex}
\label{fig4}
\end{figure}

Joints between consecutive subsegments of a corridor may be viewed as OR vertices with two inputs, shaped like in Figure~\ref{fig3}, but without the corridor coming from the left.

Finally, Figure~\ref{fig5} shows the upper area of the construction, reached by the distinguished edges $e_a$ and $e_b$ (respectively,~(a) and~(b)), and by all the pipes~(c). The guard in~(d) can cap all the pipes, one at a time, and its purpose is to clear the left part of the target region, while the small rectangle~(e) on the right will be cleared by the guard in~(f). The two pipes~(g) implement additional OR vertices with two inputs, and prevent~(d) and~(f) from acting, unless the respective distinguished edges are in their target orientations. Nook~(h) will contaminate part of the target region, unless~(d) is aiming down. Nooks~(i) prevent area~(e) from staying clear whenever guard~(f) is not aiming up. The guard in~(j) separates the two parts of the target region with its laser, so that they can be cleared in two different moments.

\begin{figure}[h]
\centering
\includegraphics[width=\linewidth]{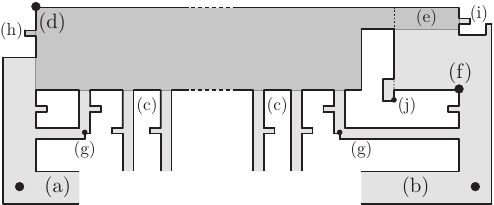}
\caption{Target region}
\label{fig5}
\end{figure}

Suppose $G$ is a solvable instance of \EEANCL. Then we can ``mimic'' the transition from configuration $A$ to configuration $B$ (see Section~\ref{sectionncl}) by turning subsegment guards. Specifically, if edge $e=\{u,v\}$ in $G$ changes its orientation from $u$ to $v$, then all the subsegment guards in the corridor corresponding to $e$ turn their lasers around, one at a time, starting from the guard closest to $u$. Before this process starts, each pipe has one end capped by some subsegment guard, and in particular pipe~(g) on the left of Figure~\ref{fig5} is capped by the guard in~(a). Hence, guard~(d) is free to turn and cap all the pipes one by one, stopping for a moment to let each pipe's internal guard clear the pipe itself (which now has both ends capped) and cover its nook (see Figure~\ref{fig2}). As a result, the left part of the target region can be cleared by rotating~(d) clockwise, from right to down. Then the subsegment guards start rotating as explained above, until configuration $B$ is reached. If done properly, this keeps all the pipes capped and clear, thus preventing the left part of the target region from being recontaminated. (Note that it makes a difference whether we turn a subsegment guard clockwise or counterclockwise: sometimes, only one direction prevents the recontamination of the pipe that the guard is capping.) When $B$ is reached, guard~(f) can turn up to clear~(e) and finally solve our \PSSP instance.

Conversely, suppose that $G$ is not solvable. Observe that rectangle~(e) in Figure~\ref{fig5} has to be cleared by guard~(f) as a last thing, because it will be recontaminated by nooks~(i) as soon as~(f) turns away. On the other hand, whenever a pipe has both ends uncapped by external guards, some portion of the target region necessarily gets recontaminated by some nook, regardless of where the pipe guard is aiming its laser. But guard~(d) can cap just one pipe at a time and, while it does so, nook~(h) keeps some portion of the target region contaminated. Thus, the entire process must start from a configuration $A$ in which all the pipes' lower ends are simultaneously capped by subsegment guards, and guard~(d) is free to turn (i.e., $e_a$ is in its target orientation). From this point onward, no pipe's lower end may ever be uncapped (i.e., legality must be preserved), otherwise the target region gets recontaminated, and the process has to restart. Finally, a configuration $B$ must be reached in which guard~(f) is free to turn up (i.e., $e_b$ is in its target orientation). By assumption this is impossible, hence our \PSSP instance is unsolvable.
\end{proof}

By putting together Lemma~\ref{lemma1} and Lemma~\ref{lemma2}, we immediately obtain the following:

\begin{theorem} \label{maintheorem}
Both \PSSP and its restriction to orthogonal polygons are strongly \PSPACE-complete. \hfill $\square$
\end{theorem}

The term ``strongly'' is implied by the fact that all the vertex coordinates generated in the \PSPACE-hardness reduction of Lemma~\ref{lemma2} are numbers with polynomially many digits (or can be made so through tiny adjustments that do not compromise the validity of the construction).

\section{Convexifying the target region} \label{convexregion}

We can further improve our Theorem~\ref{maintheorem} by making the target region in Lemma~\ref{lemma2} rectangular.

Our new target region has the same width as the previous one, and the height of rectangle~(e) in Figure~\ref{fig5}. In order for this to work, we have to make sure that some portion of the target region is ``affected'' by each contaminated pipe that is not capped by guard~(d), no matter where \emph{all} the pipe guards are oriented. To achieve this, we make pipes reach the upper area of our construction at increasing heights, from left to right, in a staircase-like fashion.

\begin{figure}[h]
\centering
\includegraphics[width=\linewidth]{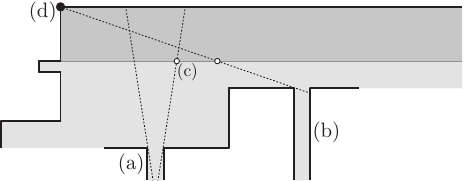}
\caption{Rectangular target region}
\label{fig6}
\end{figure}

Assume we already placed pipe~(a) in Figure~\ref{fig6}, and we need to find the correct height at which it is safe to connect pipe~(b). First we find the rightmost intersection~(c) between a laser emanating from the pipe guard of~(a) and the lower border of the target region. Then we set the height of pipe~(b) so that it is capped by guard~(d) when it aims slightly to the right of~(c). This is always feasible, provided that pipes are thin enough, which is obviously not an issue.

After we have set all pipes' heights from left to right, the construction is complete and the proof of Lemma~\ref{lemma2} can be repeated verbatim, yielding:

\begin{theorem}
Both \PSSP and its restriction to orthogonal polygons with rectangular target regions are strongly \PSPACE-complete. \hfill $\square$
\end{theorem}

\section{Further research}

There are several promising directions for future research. We suggest a few.

We could simplify \PSSP by asking if there exists a neighborhood of a given point, no matter how small, that is clearable. Let this problem be called $\PSSP^\star$. In contrast with \PSSP, here we do not have a polygonal target region, but we are interested just in the surroundings of a point. It is easy to show that $\PSSP^\star \preceq_{\P} \PSSP$: clearing a small-enough neighborhood of a point is equivalent to clearing the cells whose topological closure contains the point (cf.~the proof of Lemma~\ref{lemma1}). The author proved in~\cite{thesis} that a 3-dimensional version of $\PSSP^\star$ is \PSPACE-hard, even restricted to orthogonal polyhedra. Our question is whether $\PSSP^\star$ is \PSPACE-hard (hence \PSPACE-complete) for 2-dimensional polygons, as well.

Similarly, we may investigate the complexity of \PSSP on other restricted inputs, such as simply connected polygons, or target regions coinciding with the whole environment. The latter is in fact \SSP, whose complexity has been mentioned in Section~\ref{sectionintro} as an interesting long-standing open problem. Although the author proved that a 3-dimensional version of \SSP is \NP-hard~\cite{viglietta2}, determining the true complexity of either version still seems a deep problem. Recall that, in our \PSPACE-hardness reduction of Lemma~\ref{lemma2}, we repeatedly used regions that are visible to no guard, and hence can never be cleared. As a matter of fact, this is a remarkably effective way to force the recontamination of other areas whenever certain conditions are met. However, this expedient is of no use in a reduction for \SSP (trivially, if the guards cannot see the whole polygon, they cannot search it), and cleverer tools have to be devised for this problem.

Other interesting variations of \SSP involve the addition of new environmental elements, such as \emph{mirrors}, which specularly reflect lasers; \emph{transparent walls}, which can be traversed by lasers but not by the intruder; and \emph{curtains}, which can be traversed by the intruder but block lasers. To the best of our knowledge, none of these elements has ever been studied in connection with \SSP.

\small
\bibliographystyle{abbrv}

\end{document}